\documentclass[journal,10pt,twocolumn]{IEEEtran}

\usepackage{floatflt}

\usepackage{cite}

\usepackage{graphicx}
\usepackage{verbatim}

\usepackage{psfrag}

\usepackage{subfigure}

\usepackage{url}

\usepackage{stfloats}

\usepackage{amsmath,amsthm}
\usepackage[longend,ruled]{algorithm2e}

\newtheorem{definition}{Definition}
\newtheorem{lemma}{Lemma}

\newtheorem{property}{Property}

\newcommand{\mat}{\boldsymbol}
\newcommand{\TRUE}{true}

\begin{document}

\title{Maximizing Energy Efficiency in Multiple Access Channels by Exploiting Packet Dropping and Transmitter Buffering}

\author{M.\ Majid~Butt,~\IEEEmembership{Member,~IEEE,}
Eduard~A.~Jorswieck,~\IEEEmembership{Senior~Member,~IEEE,}
 Bj\"{o}rn Ottersten,~\IEEEmembership{Fellow,~IEEE}

\thanks{The material in this paper has been presented in part in PIMRC, September, 2013, London, UK \cite{majid_PIMRC_13}.}

\thanks{M. Majid Butt is with Computer Science and Engineering Department, Qatar University. E-mail: majid.butt@ieee.org. Eduard A. Jorswieck is with Department of Electrical Engineering and Information Technology, Dresden
University of Technology, Germany. Email: eduard.jorswieck@tu-dresden.de. Bj\"{o}rn Ottersten is with interdisciplinary center for security, reliability and trust, University of Luxembourg. Email: bjorn.ottersten@uni.lu.}
\thanks{This work was made possible by an NPRP grant 5-782-2-322 from the Qatar National Research Fund (a member of Qatar Foundation). The statements made herein are solely the responsibility of the authors.}
}

\maketitle

\IEEEpeerreviewmaketitle
\begin{abstract}
Quality of service (QoS) for a network is characterized in terms of various parameters specifying packet delay and loss tolerance requirements for the application. The unpredictable nature of the wireless channel demands for application of certain mechanisms to meet the QoS requirements. Traditionally, medium access control (MAC) and network layers perform these tasks. However, these mechanisms do not take (fading) channel conditions into account. In this paper, we investigate the problem using cross layer techniques where information flow and joint optimization of higher and physical layer is permitted. We propose a scheduling scheme to optimize the energy consumption of a multiuser multi-access system such that QoS constraints in terms of packet loss are fulfilled while the system is able to maximize the advantages emerging from multiuser diversity.
Specifically, this work focuses on modeling and analyzing the effects of packet buffering capabilities of the transmitter on the system energy for a packet loss tolerant application. We discuss low complexity schemes which show comparable performance to the proposed scheme. The numerical evaluation reveals useful insights about the coupling effects of different QoS parameters on the system energy consumption and validates our analytical results.
\end{abstract}
\section{Introduction}
The QoS parameters like throughput, latency, packet loss rate etc., characterize the behavior of network traffic. Specifically, there are some strict hard requirements in terms of worst case behavior for multimedia traffic like minimum throughput and maximum tolerable packet delay, which need to be fulfilled to maintain the quality of experience (QoE) of the application. At the same time, system energy efficiency has emerged as one of the key performance indicators for the wireless network \cite{Zia:2011,Fehske:11,Han:2013}. For the system design, QoS parameters can be treated as degrees of freedom (DoF) to achieve high system level energy efficiency. If the application is loss and delay tolerant, the DoFs can be exploited to maximize the system energy efficiency.

In the literature, energy efficient scheduling has been discussed in different settings for delay limited systems \cite{fu:06,majid_WCL12,Neely,Gallager}. The authors in \cite{Huang:2004} propose a scheme which schedules the transmission of multimedia packets in such a way that all the users have a fair share of packet loss according to their QoS requirements, and maximizes the number of the served users under the QoS constraints.
The author in \cite{Neely2009} addresses the importance of packet dropping mechanisms by energy point of view. Traditionally, average packet drop rate is considered to be one of the most important parameters for system design \cite{Karmokar,Bettesh:2006}. However, QoE for the application, specifically multimedia streaming, depends on the other characteristics of packet dropping. Average packet drop parameter characterizes the behavior of the application on long term basis only. In multimedia applications, short term behavior dictates the QoE. For example, consider a scenario where the average packet drop rate $\theta_{\rm tar}$ is quite small but a large number of packets are dropped successively due to the deeply faded wireless channel (called bursty packet loss). In spite of fulfilling an average packet drop rate guarantee (on long term basis), the users will experience a jitter in the perceived QoE (for a multimedia application). Thus, QoS must also be defined in terms of maximum number of packets allowed to be dropped successively in addition to the average packet drop probability. This additional parameter characterizing the pattern of the dropped packets is termed \emph{continuity constraint parameter} $N$ \cite{majid_TWC:13}. Packet scheduling constrained by average packet drop rate and maximum successive packet drop belongs to a class of sequential resource allocation problems, known as Restless Multiarmed Bandit Processes (RMBPs) \cite{Whittle:1988}. This problem has been addressed for Asynchronous Transfer Mode (ATM) networks in \cite{Lee:1994}. The authors in \cite{Lee_pat:1997} discuss a similar problem and an optimal dropping scheme with the objective to minimize/maximize the packet drop gap is proposed. A useful analytical framework is discussed in \cite{Fanqqin:2013} to dimension the packet loss burstiness over generic wireless channels and a new metric to characterize the packet loss burstiness is proposed.

Traditionally, such problems are handled at upper layers of communications through link adaptation or automatic repeat request (ARQ) mechanisms. However, bringing this information to the physical layer design shows significant merits as the information can effectively be used for opportunistic scheduling purposes.
The work in \cite{majid_TWC:13} proposes an opportunistic scheduling scheme which exploits the DoF available through continuity constraint and average packet drop parameters and aims at minimizing the average system energy. The work characterizes the effects of the $\theta_{\rm tar}$ and $N$ on the system energy consumption. However, the proposed scheme does not allow buffering of data packets which is an integral part of the resource allocation mechanisms.

This work generalizes the framework in \cite{majid_TWC:13} for the case when buffering of packets is allowed for a finite number of time slots on the transmitter side. This additional DoF poses new challenges in terms of modeling and analysis of the problem because a buffer provides multiple opportunities to exploit multiuser diversity as compared to a single opportunity in real time traffic. In addition to the QoS parameters $\theta_{\rm tar}$ and $N$, the size of the buffer $B$ provides another trade-off for energy efficiency. It should be noted that in contrast to the conventional system design goal of dropping of data packet as a consequence of not being able to provide required rate to the users, our approach encourages dropping of packets to optimize the system energy consumption as long as the QoS parameters allow.

We investigate the energy efficiency of the system constrained by the coupled packet drop (QoS) parameters. The main contributions of this work are as follows:
 \begin{itemize}
   \item We allow buffering at transmitter side which allows multi-packet scheduling as compared to a single packet scheduling in \cite{majid_TWC:13}. The buffering effect to exploit channel diversity has been well studied in literature, but average and successive packet (bursty) loss constraints have not been investigated simultaneously over fading channels for the buffering system. The modeling of successive packet loss constraint requires not only to decide how many (average) packets need to be transmitted, but which of them are more significant with respect to QoE.
   \item We propose a novel scheduling scheme which takes into account channel distribution, packet loss characteristics and maximum delay limitations for a packet. This generalized framework is more complex due to involvement of an additional DoF, but provides better results in terms of energy efficiency as demonstrated through asymptotic case analysis and numerical evaluation.
   \item We investigate and quantify the effect of buffer size on system energy mathematically and characterize the dominating regions for each system parameter (e.g., buffer size, $\theta_{\rm tar}, N$) in terms of energy efficiency. We show that increasing buffer size indefinitely does not help to increase energy efficiency of the system for a fixed $N$ and $\theta_{\rm tar}$.
   \item The complexity of the proposed scheme is quite high for large buffer size. Therefore, we propose and analyze the low complexity solutions. The energy loss due to sub-optimality is evaluated numerically, which reveals the interesting result that the optimal and low complexity schemes show comparable energy performance.
 \end{itemize}

The rest of the paper is organized as follows. Section \ref{sect:system model} introduces the system model and fundamental assumptions. We discuss and analyze the proposed scheme in Section \ref{sect:scheme} while the optimization problem is formulated in Section \ref{sect:optimization}. The effect of buffer size on system energy is characterized mathematically in Section \ref{sect:buffer_size}. We discuss low complexity schemes in Section \ref{sect:suboptimal} and compare them with our proposed scheme. We provide numerical evidence of the tradeoff between energy, data loss and buffer size for our schemes in Section \ref{sect:results} and conclude with the main contributions of the work in Section \ref{sect:conclusions}.

\section{System Model}
\label{sect:system model}
In this paper, we follow the system model used in \cite{Ralf1,majid_TWC:13}. We consider a multiple-access system with $K$ users randomly placed within a certain area.
The system is able to provide a certain fraction of the total data rate to each user. Every
scheduled packet for a user $k$ has normalized size $R_k=\frac{C}{K}$ where $C$ denotes the spectral efficiency of the
system.

We consider an uplink scenario where time is slotted such that each user $k$ experiences a channel gain $h_k(t)$ in a time slot $t$. The channel
gain $h_k(t)$ is the product of path loss $s_k$ and small-scale
fading $f_k(t)$. The path loss is a function of
the distance between the transmitter and the receiver and remains constant within the time scale considered in this work. Small-scale fading depends on the scattering environment. It changes
from slot to slot for every user and is independent and identically
distributed (i.i.d.) across both users and slots, but remains constant during the time span of a single time slot. The multi-access channel is described by the input vector (X) and output vector (Y) relation as,
\begin{equation}
Y(t)=\sum_{k=1}^K \sqrt{h_k(t)}X(t)+ Z(t)
\end{equation}
where $Z$ represents additive i.i.d. complex Gaussian
random variable with zero mean and unit variance.
The channel
state information (CSI) is assumed to be known at both transmitter and receiver sides.

The continuity constraint requires us to allow scheduling of multiple users simultaneously in the same time slot. If only a single user is scheduled per time slot, the continuity constraint cannot be satisfied without allowing outage when multiple users have already dropped $N$ packets. The analysis of the scheme is based on asymptotic user case and therefore, scheduling of very large number of simultaneous users is desirable. We use superposition coding and successive interference cancelation (SIC) mechanism for successful transmission of data (rate) of simultaneously scheduled users\footnote{Theoretically, there is no limit on the number of users scheduled simultaneously if there is no power constraint. In practice, the users have peak or average power constraints and apply finite modulation and coding schemes which limit the number of users scheduled simultaneously.}. Treating the other users as interference, we end up with a Gaussian channel for which modulation and coding schemes are well researched, e.g., \cite{Forney:1998}.

\begin{figure*}
\centering
\includegraphics[width=4.0in]{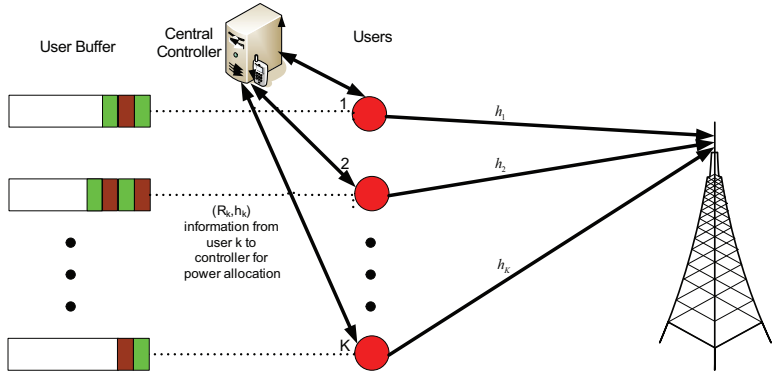}
  \caption{The system model settings for the problem. The users make scheduling decisions independently based on the channel conditions and provide rate and channel state information to a central controller which provides the required power allocation information to the users for superposition coding.}
  \label{fig:system_diagram}
  \vspace{-0.3cm}
\end{figure*}

Let ${\mathcal{K}}$ denote the set of users
to be scheduled and $\Phi_1\cdots\Phi_{|\mathcal{K}|}$ be the
permutation of the scheduled user indices that sorts the channel gains in increasing order, i.e.\
$h_{\Phi_1}\le \cdots \le h_{\Phi_k}\le \cdots \le
h_{\Phi_{|{\mathcal K}|}}$. Then, the energy of the scheduled user
$\Phi_k$ with rate $R_{\Phi_k}$, is given by,
\cite{Tse3,Ralf1}
\begin{equation}\label{eqn:power}
 E_{\Phi_k} = \frac{Z_0}{h_{\Phi_k}} \left({2^{\sum_{i\leq k}R_{\Phi_i}}-2^{\sum_{i<k}R_{\Phi_i}}}\right)
\end{equation}
where $Z_0$ denotes the noise power spectral density. Fig.~\ref{fig:system_diagram} shows the block diagram for the system settings where a central controller is responsible for providing the required power level information to each user for superposition coding.

\section{Modeling of the Proposed Scheme}
\label{sect:scheme}
A constant arrival of a single packet with normalized size $\frac{C}{K}$ is assumed for simplicity. However, the scheme is not restricted to this assumption as a random packet arrival process can be modeled as a constant arrival process where multiple arrived packets in the same time slot are merged as a single packet with random packet size following the framework in \cite{majid_TWC:13,majid_WCL12}. The packet arrival occurs at the start of a time slot and the scheduling is performed afterwards taking into account the newly queued packet. All the arriving packets are queued sequentially, i.e., the oldest arrived packet is the head of line (HOL). If a single packet has to be scheduled or dropped, it has to be the HOL packet. Note that successive packet drop constraint inspires us to buffer and drop the packets sequentially (as compared to any random queuing strategy) because it is essential to maintain a sequence of the packets in the transmitter buffer. The newly arrived packet is the pointer to indicate how essential it is to transmit or drop a packet in relation to successive packet drop constraint while the scheduling of the other buffered packets is essential to maintain a ceratin average packet drop rate.

The continuity constraint and buffer size parameters for a user $k$ are denoted by $N$ and $B$, respectively; and assumed to be identical for all the users\footnote{The framework can be generalized to non-identical $N$ case for individual users where $N$ follows a probability distribution, but this is out of scope of this work.}.
The variables $d\leq N$ and $b\leq B$ denote the number of successively dropped and buffered packets for a user $k$ at time $t$, respectively. A packet arriving at time $t$ is not dropped immediately if not scheduled but buffered up to $B$ time slots and dropped then (if still not scheduled).

We use Markov decision process (MDP) to model and analyze the scheme which is a useful tool due to dependency of the dropped packets in relation to the successive packet drop sequences. The state space of the user is defined as
\begin{equation}
\Lambda=\bigl\{(\mathcal{D},\mathcal{B});\mathcal{D}\in\{0,\dots,N\},\mathcal{B}\in\{0,\dots B\}\bigr\}
\end{equation}
where $\mathcal{D}$ and $\mathcal{B}$ denote the state space for successive packet drop and buffer states, respectively. Then, the state is defined as a composite variable $p\in\Lambda$ by the summation of the number of (already) \emph{successively} dropped and (already) buffered packets at time $t$ such that
\begin{equation}
p=d+b~.
\end{equation}
At the start of the Markov process ($p=0$), the packet is not dropped if not scheduled as packets can be buffered for $B$ time slots resulting in $d=0$ and $p=b$ for $p \leq B$. When the buffer is completely filled with packets, the unscheduled HOL packet is dropped onwards. Note that, the dropping operation is limited to a single packet as this is enough to make room for the newly arrived\footnote{The newly arrived packet waits in a separate \emph{temporary} buffer momentarily before the scheduling decision as it arrived at the start of the time slot.} packet at time $t$. Thus, the variable $d$ increases and $b$ is fixed to $B$ for $p> B$. The maximum number of states in our Markov chain is $B+N+1$ where $M=B+N$ denotes the termination state.

Let $\alpha_{pq}$ denote the transition probability from a state $p$ to $q$ with $S_t$ representing the state at time $t$. Furthermore, we denote transition probabilities associated with the scheduling, buffering and dropping decisions by the notation ${\alpha}_{pq}^s$, ${\alpha}_{pq}^b$ and ${\alpha}_{pq}^d$, respectively.
We define $\alpha_{pq}$ as
\begin{eqnarray}
\alpha_{pq} = {\rm Pr}(S_{t+1}=q|S_t=p)
=\begin{cases}
{\alpha}_{pq}^s & \forall p,q\leq \min(p,B)\\
{\alpha}_{pq}^b & p< B,q=p+1\\
{\alpha}_{pq}^d&p\geq B,q=p+1\\
0& \mbox{otherwise}
\end{cases},
\end{eqnarray}
where possible values of the states form the finite state space $\Lambda$ of
the MDP.

We define a scheduling threshold.
\begin{definition}[Scheduling Threshold $\kappa_{pq}$] It is defined as the minimum small scale fading value $f_k$ required to make a state transition from state $p$ to $q$ such that
\begin{equation}
{\alpha}_{pq}^s =
{\rm Pr}\bigl(\kappa_{pq}<f_k\leq\kappa_{p(q-1)}\bigr)\quad 0\leq q\leq \min(p,B),
\label{eqn:alpha1}
\end{equation}
where $\kappa_{p0^-}$ is defined to be infinity with $S_{0-}$ denoting a dummy state before $S_0$.
\end{definition}
Thus, the small scale fading must be greater than $\kappa_{pq}$ for the process to enter in state $q$ where fading vector is quantized in non-overlapping intervals. The threshold definition uses fading instead of channel gain to avoid near-far effect which gives unfair advantage to the users near the base station.

\subsection{The Proposed Scheduling Scheme}

\begin{figure*}[!t]
 \centering
\subfigure[N=2, B=1]
 {\includegraphics[width=3.0in]{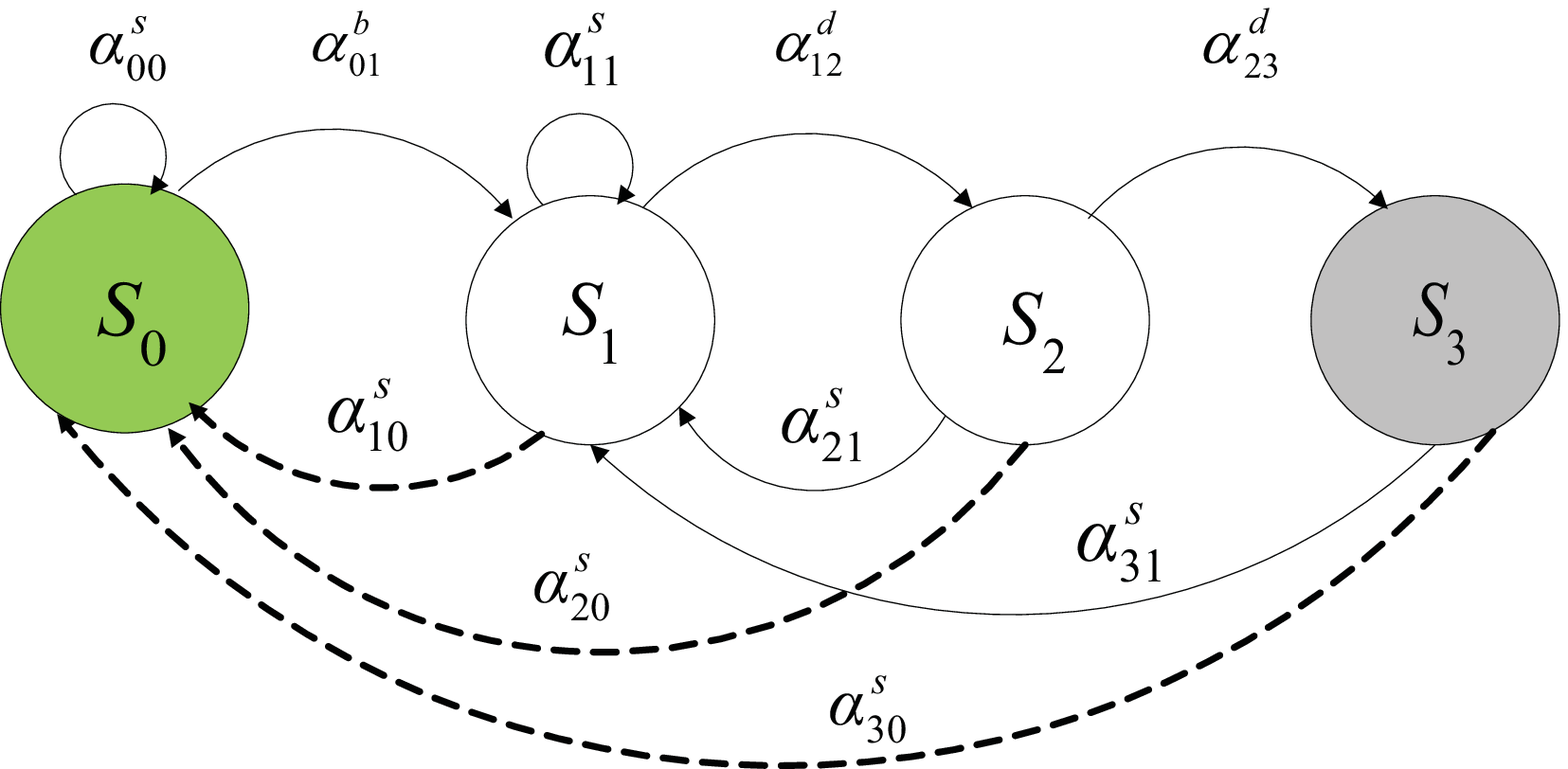}}~~
\subfigure[N=1, B=2]
 {\includegraphics[width=3.0in]{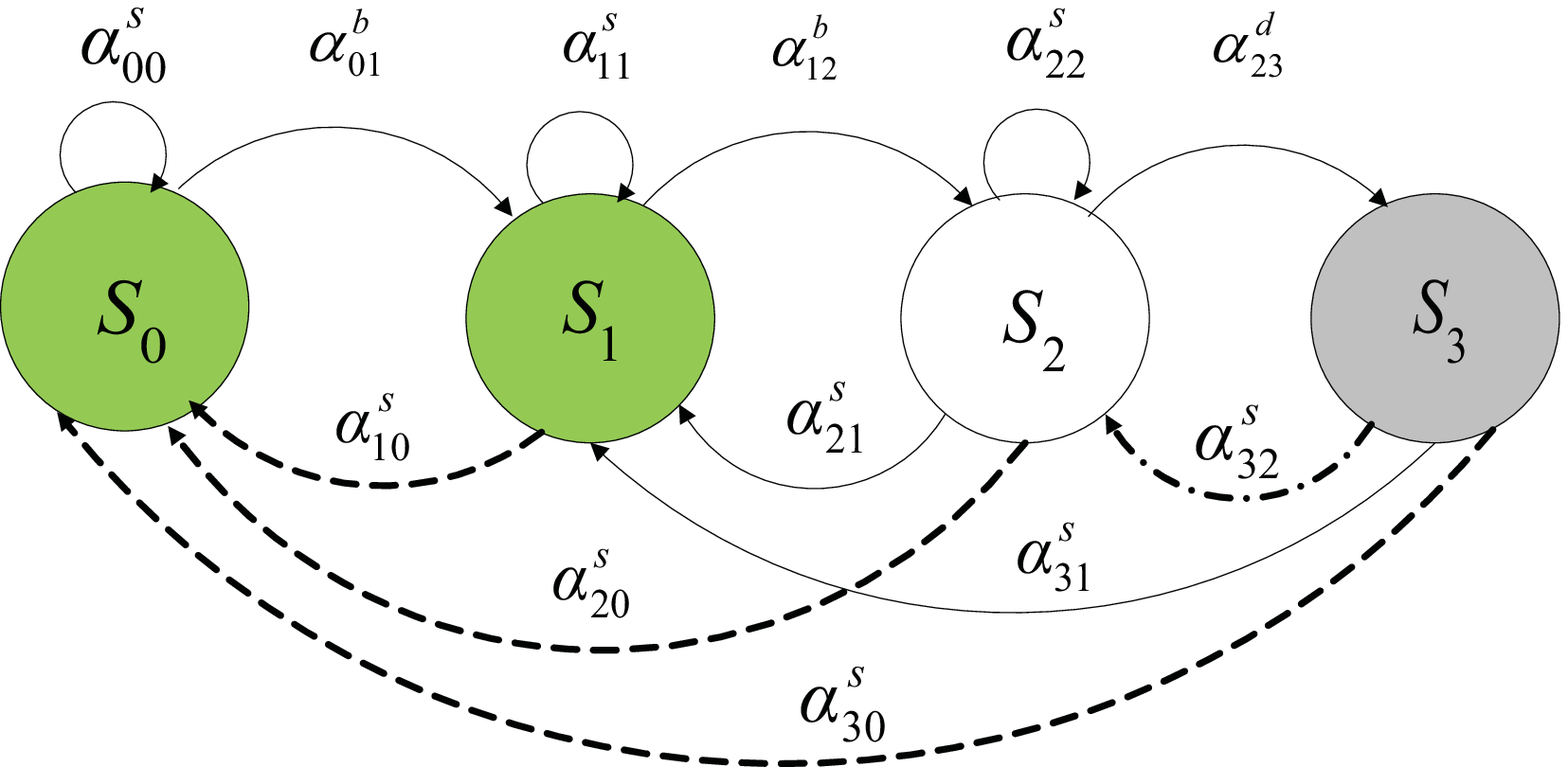}}
  \caption{State transition diagrams for different buffer size and continuity constraint parameters. The green colored states represent the buffering states while grey state is the $M^{\rm th}$ state. The superscript notation with transition probabilities depict the actions associated with the transitions.}
  \vspace{-0.3cm}
  \label{fig:state_diagram}
\end{figure*}

We assume an infinitely large number of users in the system for the analysis purpose. This assumption makes the analysis of the scheme tractable which is hard otherwise due to multiuser interactions. In the large user limit, the scheduling decisions of the users decouple and the multiuser system can be modeled as a single user system following the work in \cite{viswanath:01,majid_TWC:13}. Every user makes his own scheduling decision independent of the other users.

The purpose of the scheduling scheme is to maximize the use of available fading conditions by scheduling as many packets as possible. Thus, the fading vector is quantized in such a way that the discrete set of state-dependent scheduling thresholds determine the intervals for the optimal scheduling decisions and the size of this vector equals the number of packets available for scheduling in a state $p$.

In a state $p\geq q$, the scheduler makes a state transition to state $q$ for fading $f_k$ such that
\begin{equation}
\kappa_{pq}<f_k\leq\kappa_{p(q-1)}, \quad 0\leq {q}\leq \min(p,B)~.
\label{eqn:decision}
\end{equation}
For a state transition ${\alpha}_{pq}^s$ with $q\leq \min(p,B)$, the number of the scheduled packets is given by
\begin{equation}
L(p,f) = \min(p,B)-q+1~,
\label{eqn:packets}
\end{equation}
where $q$ is determined uniquely by (\ref{eqn:decision}).
Obviously, a user can only schedule as many packets as buffered. Thus, the maximum scheduled packets for a state $p<B$ are limited to $p-q+1$ (due to constant arrival model) while they are fixed to $B-q+1$ for $p\geq B$.

Note that scheduling of packets starts with the HOL packet and ends with the most recently arrived packet. Given a state $p$, (\ref{eqn:decision}) chooses $q$ which maximizes the average system reward. To meet the continuity constraint with probability one,
$\kappa_{MB}$ is set to zero to allow transmission of the HOL packet in state $M$, similar to the approach used in \cite{jindal_2009}.

We deduce the following properties of the proposed scheduling scheme from (\ref{eqn:decision}).
\begin{property}
The next state (in case of scheduling) is limited by the minimum of $p$ and $B$. If $p\leq B$, $q$ cannot exceed $p$, otherwise it is limited to $B$.
\end{property}
Thus, up to $\min(p,B)$ buffered packets and one newly arrived packet can be scheduled depending on small scale fading in a state $p$. As the quantity $\min(p,B)$ is used extensively in the analysis, we denote it by $\mu$ in the rest of the paper.
\begin{property}
Scheduling thresholds in a state $p$ follow the monotonic decrease property that
\begin{equation}
\kappa_{pq}<\kappa_{p(q-1)}~,\quad\forall p,0\leq q< \mu ~.
\label{eqn:P1}
\end{equation}
\end{property}
\begin{property}
For a state transition resulting in the scheduling of the same number of packets,
\begin{equation}
\kappa_{pq}\leq\kappa_{(p-1)(q-1)} ~,\quad\forall p,0\leq q< \mu ~.
\label{eqn:prop3}
\end{equation}
\end{property}
The closer to termination state, the smaller the thresholds and the larger the energy expenditure. This follows from the classical optimal stopping theory formulation where expected incentive in waiting decreases monotonically \cite{jindal_2009}.

If $f\leq \kappa_{p\mu}$, no scheduling occurs. In this case, the next state $q$ equals $p+1$ but a packet can be dropped or buffered depending on the conditions in (\ref{eqn:alpha2}) and (\ref{eqn:alpha3}) such that
\begin{equation}
{\alpha}^b_{p(p+1)}= {\rm Pr} (f\leq\kappa_{pp}) = 1 - \sum_{q=0}^p{\alpha}_{pq}^s~,\qquad p<B~.
\label{eqn:alpha2}
\end{equation}
$\kappa_{pp}$ denotes the minimum threshold for self transition to schedule at least one packet.
If $p<B$, the unscheduled HOL packet is buffered with the option that it can be scheduled in one of the $B-p$ time slots in future.
\begin{equation}
{\alpha}^d_{p(p+1)}= {\rm Pr} (f\leq\kappa_{pB}) = 1 - \sum_{q=0}^B{\alpha}_{pq}^s~, \qquad p\geq B,
\label{eqn:alpha3}
\end{equation}
where $\kappa_{pB}$ denote the minimum thresholds to make transition to state $B$ for scheduling at least one packet.
If $p>B$, the unscheduled HOL packet has to be dropped as the buffer is already full. The best option by continuity constraint point of view is to drop HOL packet to make room for the newly arrived packet.

The state transition probabilities are computed as a function of scheduling thresholds via (\ref{eqn:alpha1}), (\ref{eqn:alpha2}) and (\ref{eqn:alpha3}), and depend on the fading distribution and QoS parameters ($N,B,\theta_{\rm tar}$). The state transition probabilities constitute the state transition probability matrix $\mat{Q}$ for an MDP.

To explain the computation of matrix $\mat{Q}$, consider the state diagram in Fig. \ref{fig:state_diagram} for different values of $N$ and $B$ parameters.
The corresponding matrix $\mathbf{Q}$ for a system with $N=1$ and $B=2$ is given by
\begin{equation}
\mathbf{Q} = \left( \begin{array}{cccc}
{\alpha}^s_{00} & {\alpha}^b_{01}&0  &0 \\
{\alpha}^s_{10} & {\alpha}^s_{11} & {\alpha}^b_{12} &0 \\
{\alpha}^s_{20} & {\alpha}^s_{21}  & {\alpha}^s_{22} & {\alpha}^d_{23}\\
{\alpha}^s_{30} & {\alpha}^s_{31}  & {\alpha}^s_{32} & 0\\
\end{array} \right) \label{eqn:Drop_STM}.
\end{equation}
Appendix \ref{app:k_relation} shows the general relation between matrix $\mathbf{Q}$ and scheduling thresholds for a small scale fading distribution $p_f(y)$.

It should be noted that the number of states in an MDP are the same for the parameter sets $N=2, B=1$ and $N=1,B=2$ but the transition probability matrix $\mathbf{Q}$ differs and captures the effect of each parameter on the system energy. The parameter set $N=2, B=1$ requires optimization of 2 thresholds per state for $p\geq 1$ while the parameter set $N=1, B=2$ requires $3$ thresholds per state for $p\geq 2$. We evaluate the energy efficiency of both cases numerically in Section \ref{sect:results}.

\section{Optimization Problem Formulation}
\label{sect:optimization}
The number of scheduled packets are considered virtual users (VU) for the analysis purpose.
It is known that the average energy consumption of the system per transmitted information bit at
the large system limit $K \to \infty$ is given by \cite{majid_TWC:13,Ralf1}
\begin{equation}
\label{eqn:energy_function}
\left(\frac{E_{\rm b}}{N_0}\right)_{\rm
sys} =  \log(2) \int\limits_0^\infty \frac {2^{C \,{\rm P}_{h,\rm
VU}(x)}}x\, {\rm dP}_{h,\rm VU}(x)
\end{equation}
where ${\rm P}_{h,\rm VU}(\cdot)$ denotes the cumulative distribution function (cdf) of the fading of the scheduled VUs. It comprises of the small scale fading and the path loss components of the VUs. However, in the large system limit, the state transitions depend only on the small scale fading distribution as the path loss for VUs follows the same distribution as the path loss of the users.

The optimization problem is to minimize system energy for the target average and successive packet drop constraints. We formulate the problem with the help of MDP model described in Section \ref{sect:scheme} and optimize the state transition probabilities (or resulting $\mat{Q}$) offline.

Consider the following optimization problem:
\begin{eqnarray}
\label{eqn:objective}
&\min_{\mathbf{Q}\in \Omega} \left(\frac{E_{\rm b}}{N_0}\right)_{\rm{sys}}&\\
&\mbox{s.t.} \begin{cases}\mathcal{C}_1: 0\leq\sum_{m=0}^{\mu}\alpha_{pm}\leq 1 &   0 \leq \alpha_{pm} \leq 1, 0\leq p\leq M\\
\mathcal{C}_2: \theta_r\leq \theta_{\rm tar} &  \mathbf{Q}\in \Omega\\
\mathcal{C}_3:\sum_{q=0}^M \alpha_{pq}=1& 0\leq p \leq M\\
 \mathcal{C}_4:B+N=M & B<\infty, N<\infty
\end{cases},&
\label{eqn:optimization}
\end{eqnarray}
where $\Omega$ denotes the set of permissible matrices for $\mathbf{Q}$. Note that $\theta_{\rm tar}$ is the target (maximum) average packet drop rate while $\theta_r$ is the average packet drop rate resulting from our proposed scheme and any fixed $\mathbf{Q}$ according to (\ref{eqn:optimization}) such that
\begin{eqnarray}
\theta_r = \sum_{p=B}^{M-1}\alpha_{p(p+1)}\pi_p
= \sum_{p=B}^{M-1} \Bigl(1-\sum_{m=0}^{\mu} \alpha_{pm}\Bigr)\pi_p~.
\label{eqn:drop_cons2}
\end{eqnarray}
Equation (\ref{eqn:drop_cons2}) results by combining  $\mathcal{C}_1$ and $\mathcal{C}_3$ in (\ref{eqn:optimization}) while
$\pi_p$ denotes the steady state probability of the state $p$ and follows the property of homogenous Markov Chain that
\begin{equation}
 \sum_{p=0}^{M}\pi_p =1~.
\end{equation}
In our MDP model, the forward transition for the state $p\geq B$ represents the events of dropping the packet in state $p$ and the summation over the corresponding transition probabilities $\alpha_{p(p+1)}$ gives the average dropping probability $\theta_r$. The summation in (\ref{eqn:drop_cons2}) starts from state $B$ as the unscheduled packets are buffered for $p<B$.

The constraint on average packet drop rate is modeled by $\mathcal{C}_2$ in (\ref{eqn:optimization}).
The constraint on successive packet drop is modeled through $\mathcal{C}_4$ while $\mathcal{C}_1$ and $\mathcal{C}_3$ are the standard constraints on the transition probabilities. For a fixed $p$, the corresponding channel-dependent optimal scheduling thresholds can be computed from the optimized $\vec{\alpha_p^*}=[\alpha_{p0}^*,\dots \alpha_{p\mu}^*]$ using (\ref{eqn:alpha1}).

$\left(\frac{E_{\rm b}}{N_0}\right)_{\rm sys}$ in (\ref{eqn:energy_function}) requires computation of channel distribution for the scheduled users.
From our MDP model for the problem, the probability density function (pdf) of the small scale fading of the scheduled VUs is computed as a function of the optimized scheduling thresholds by
\begin{equation}
{\rm p}_{f,\rm VU} (y) = \sum\limits_{p=0}^M c_pL(p,y)\pi_p\,{\rm p}_{f}(y),\quad \kappa_{pq}<y\leq\kappa_{p(q-1)} \label{eqn:SVU_fading}
\end{equation}
where ${\rm p}_{f}(y)$ and $c_p$ denote the
small scale fading distribution and a normalization constant while $L(p,y)$ is given by (\ref{eqn:packets}). The pdf represents the conditional probability of being in state $p$ and transmitting $L(p,y)$ packets for small sale fading $y$.

The cdf of the VUs can be written as a sum of integrals
\begin{eqnarray}
\label{eqn:SVU_cdf_integrals}
{\rm P}_{f,\rm VU}(y)&=& \sum\limits_{p=0}^M c_p\pi_p
\Bigl(L(p,y)\int \limits_{\kappa_{pq}}^y {\rm p}_{f}(\xi){\rm d\xi} \\
&+& \sum_{b=1}^{\mu-q}b\int
\limits_{\kappa_{p(\mu-b+1)}}^{\kappa_{p(\mu-b)}}
{\rm p}_{f}(\xi){\rm d\xi}\Bigr)\nonumber\\
&=&\sum\limits_{p=0}^M c_p\pi_p\Bigl(
L(p,y){\rm P}_{f}(y)
-\sum_{b=0}^{\mu-q}{\rm P}_{f} (\kappa_{p(\mu-b)})\Bigr)~,
\label{eqn:SVU_cdf3}
\end{eqnarray}
since no virtual users are scheduled for $y<\kappa_{p\mu}$. The channel distribution for the scheduled VUs is computed in close form for Rayleigh fading and the path loss distribution in Appendix \ref{app:channel}.
\subsection{Heuristic Optimization}
The optimization problem is to compute a set of transition probabilities that result in minimum system energy in (\ref{eqn:energy_function}). For every state $p$, an optimal $\vec{\alpha^*}=[\alpha_{p0}^*,\dots \alpha_{p\mu}^*]$ needs to be computed. The computation of optimal $\vec{\alpha^*}$ under constraints in (\ref{eqn:optimization}) is a stochastic optimization problem and requires heuristic optimization techniques like genetic algorithms, neural networks, etc., which provide acceptable solutions with reasonable computational complexity.

We choose Simulated Annealing (SA) to compute the solution for the optimization problem. SA is believed to help avoiding local minima by probabilistically allowing a candidate configuration to be the best known solution temporarily even if the configuration is not the best available solution at that time. This is called muting. The muting occurs at a faster rate at the start of the optimization process and decreases as the process goes on. This behavior is analogous to first heating at high temperature and then cooling the object. There are many cooling schedules used in literature, e.g., Boltzmann annealing (BA) and Fast annealing (FA) schedules, etc., \cite{BA,FA}. We employ FA in this work. In FA, it is sufficient to decrease the temperature linearly in each step $j$ such that,
\begin{equation}
\label{eqn:BA} T_j = \frac{T_0}{c_{\rm sa}j+1}
\end{equation}
where $T_0$ is a suitable starting temperature and $c_{\rm sa}$ is a constant adjusted according to the requirements of the problem.

Algorithm \ref{algorithm} shows the pseudocode for application of SA to our optimization problem at a single temperature iteration. The MDP models the problem for finite $B,N$ while the matrix $\mathbf{Q}$ represents a candidate configuration for SA. In SA, one transition probability fulfilling $\mathcal{C}_1-\mathcal{C}_4$ (in (\ref{eqn:optimization})) in $\mathbf{\hat{Q}}$ is varied randomly in a single iteration and the objective function, $\left(\frac{E_{\rm b}}{N_0}\right)_{\rm
sys}$ in (\ref{eqn:energy_function}), is computed \emph{only} if the constraint $\mathcal{C}_2$ is satisfied. If the solution improves the previous best solution, the new configuration, i.e., the proposed $\mathbf{\hat{Q}}$, is selected as the current best solution, discarded otherwise. However, with a certain probability, the new matrix $\mathbf{\hat{Q}}$ can be selected as the best solution even if it does not improve the previous best solution. This helps to avoid local minima as explained earlier. After $n$ iterations, the temperature is lowered down according to (\ref{eqn:BA}) and the routine is repeated untill the temperature reaches its lower limit and the computed end solution is considered optimal. We omit the details of SA scheme here due to space limitations but the reader is referred to \cite{SA_1,SA_2} for the details of SA algorithm.

\renewcommand{\baselinestretch}{1.0}
\begin{algorithm}
\caption{SA Algorithm iterations for a single iteration at temperature T}
\KwIn{$(\mat{Q^*},E^*,T,\theta_{\rm tar}$)\;}
$E^*$= current minimum energy\;
$\mat{Q^*}$=  current probability matrix solution \;
\tcc{Generation of $n$ random $\mat{\hat{Q}}$ for a temperature $T$.}
\For{i=0 \KwTo n}{
 Generate a random $\mat{\hat{Q}}$ and compute $\theta_r$ for $\mat{\hat{Q}}$ \;
  \tcc{Evaluating $\mathcal{C}_2$}
  \If{$(\theta_r<\theta_{\rm tar}$)}{
  Compute energy $\hat{E}$ as a function of $\mat{\hat{Q}}$\;
  \If {(Muting is \TRUE)}{
        $\mat{Q^*}=\mat{\hat{Q}}$\;
        \tcc{Energy update step.}
     \If {($\hat{E}< E^*$)}{
        $E^* =\hat{E}$;

        }
     }
   }
  }
\Return{$(E^*,\mat{Q^*}$);}
\label{algorithm}
\end{algorithm}

\section{Characterization of Buffer Size}
\label{sect:buffer_size}
In this section, we characterize the effect of buffer size on system energy. For a fixed $N$, an increase in buffer size $B$ causes increase in quantization levels for (per state) fading vector. Larger the quantization levels, the better the use of fading to improve energy consumption in our scheduling scheme. Technically, larger the buffer size, the more is the waiting time for a specific packet to wait for an optimal time slot to get scheduled. This result follows directly from the finite horizon optimal stopping theory.

To characterize the effect of buffer size on our scheme, we compare two systems with equal $N$: a system with no buffering and a system with buffer size $B$. As $M=N+B$, denoting continuity constraint parameter for a zero buffer system by $\acute{N}$ implies $\acute{M}=\acute{N}$ and (\ref{eqn:drop_cons2}) reduces to
\begin{eqnarray}
\theta_r = \sum_{p=0}^{\acute{N}-1}\alpha_{p(p+1)}\pi_p~.
\label{eqn:drop_cons3}
\end{eqnarray}
Eq. (\ref{eqn:drop_cons3}) for $B=0$ is identical to (\ref{eqn:drop_cons2}) for a system with $B>0$ in terms of number of states summed over as $B$ is added in both upper and lower limits to get (\ref{eqn:drop_cons2}). The only difference between the two systems emerges from the fact that a proportion of the packets for a buffering system is scheduled before reaching state $B$ in MDP. Let $\lambda_B$ denote the effective rate of unscheduled packets entering state $B$ (as compared to one in case of no buffering) and given by
\begin{eqnarray}
\lambda_B=\alpha_{(B-1)B}\pi_{B-1} = 1-\sum_{p=0}^{B-1}\sum_{q=0}^p\alpha_{pq}\pi_p
\label{eqn:drop_cons4}
\end{eqnarray}
The increased energy efficiency in the system with buffering is due to the packets scheduled in first $B$ (0 to $B-1$) states as explained in the following.
\begin{enumerate}
  \item {The packets scheduled in first $B$ states help to reduce the number of packets to be scheduled in the remaining $N$ states while $\theta_{\rm tar}$ is the dropping rate for the packets entering the system. As a fraction $1-\lambda_B$ of the packets is scheduled before state $B$, a larger proportion of the remaining packets can be dropped in remaining $N$ ($B$ to $M$) states to meet $\theta_{\rm tar}$ as compared to the system without buffering which increases DoF and helps to reduce system energy. Thus, due to scheduling in additional $B$ states, the resulting system energy for the dropping probability $\theta_{\rm tar}$ and buffer size $B$ is equivalent to the system energy for the case $B=0$ and dropping probability $\theta_{\rm tar}/\lambda_B$. In numerical results, we quantify that the energy gain due to factor $\theta_{\rm tar}/\lambda_B$ provides a lower bound for the improvement as additional gain is achieved in scheduling in first $B$ states as explained below.}
  \item {The packets scheduled in first $B$ states result in smaller average $E_b/N_0$ as compared to average $E_b/N_0$ in the remaining $N$ states. This follows from Property 3 and (\ref{eqn:prop3}). An increasing scheduling threshold implies smaller energy consumption in the scheduled packets in first $B$ states.}
\end{enumerate}

\subsection{Limiting the Buffer Size}
In practice, the continuity constraint parameter values are of the order of a few tens. Our results in Section \ref{sect:results} show that increasing the value of $B$ for a fixed $N$ results in a decrease of energy. Naturally, one would like to have a large value for $B$ to maximize the energy gain. However, $B$ cannot be increased indefinitely due to the following phenomenon.

For a fixed $N$, the system energy saturates at some $\theta_{\rm lim}=\theta_{\rm tar}$ and $\theta_{\rm tar}>\theta_{\rm lim}$ does not improve the energy efficiency (c.f. Lemma 1 in \cite{majid_TWC:13} as reviewed briefly in Appendix \ref{app:review}) where $\theta_{\rm lim}$ is the solution of (\ref{eqn:objective}) without applying $\mathcal{C}_2$ in (\ref{eqn:optimization}). In the following, we study the effect of buffer size $B$ on this particular parameter. We observe from the numerical results that increasing $B$ for a fixed $N$ shifts $\theta_{\rm lim}$ towards zero average dropping probability, i.e., increasing $B$ further does not help to increase the system energy efficiency significantly.

We deduce the following Lemma based on our evaluation.
\begin{lemma}
There exists a value of $B$ that maximizes the energy gain via buffering for a fixed $N$. It is independent of $\theta_{\rm tar}$ and denoted by $B_{m}(N)$.
\end{lemma}
\begin{proof}
We obtain $\theta_{\rm lim}=\theta_r$ from the solution of the problem solely constrained by the continuity constraint, i.e., by solving (\ref{eqn:optimization}) without constraint $\mathcal{C}_2$ for forward transition probabilities $\alpha_{p(p+1)}$ for $B\leq p\leq M-1$. As $B$ increases with respect to $N$, the steady state probabilities $\pi_p$ for $B\leq p\leq M-1$ decrease because a buffered packet gets $B$ opportunities to exploit multiuser diversity before it gets dropped if not scheduled. From (\ref{eqn:drop_cons2}), $\theta_{\rm lim}$ depends on $\pi_p$ for $B\leq p\leq M-1$ as well, i.e., $\theta_{\rm lim}(\pi_B,\pi_{B+1},...,\pi_{M-1})$. If one component $\pi_B,\pi_{B+1},...,\pi_{M-1}$ is decreasing, $\theta_{\rm lim}$ decreases for a fixed $N$, i.e., $\frac{\partial \theta_{\rm lim}(\pi_B,\pi_{B+1},...,\pi_{M-1})}{\partial \pi_l} \geq 0$ for $B \leq l \leq M-1$. Eventually $\theta_{\rm lim}\to 0$ for $B/N>>1$.
\end{proof}

Thus, parameter $N$ constrains $B$ via Lemma 1. However, this maximization is achieved at $\theta_{\rm lim}=0$ which implies that the value of $\theta_{\rm tar}$ and $N$ is irrelevant as it is not permissible to drop any packet regardless of $N$. Therefore, the value of $B<B_m(N)$ should be designed such that the DoF available from $N$ and $\theta_{\rm tar}$ can be utilized effectively to improve the system energy efficiency in conjunction with parameter $B$. 

\section{Suboptimal Scheduling Schemes}
\label{sect:suboptimal}
The computational complexity of the scheme depends on the number of quantization levels (thresholds) per state which in turn depend on buffer size.
In practice, the buffer size is of the order of a few tens to hundreds. In this case, the "Best"\footnote{Please note that we avoid using term optimal as the solution of the scheme presented in Section \ref{sect:scheme} cannot be proven optimal in the mathematical sense. To differentiate this scheme with low complexity schemes, we term the scheme as "Best" in the subsequent work.} scheduler presented in our work results in computational complexity (for the thresholds) of the order $\mathcal{O}(MB)$ and in the case $B>>N$, it becomes $\mathcal{O}(B^2)$. In the following, we propose suboptimal schedulers which reduce the complexity at the marginal energy loss. These schedulers exploit non-uniform distribution of transition probabilities in the original optimal scheme. For every state $p$, the original scheduling scheme is based on the idea of allowing scheduling of multiple packets for opportunistic use of good channels.

However, the computational complexity can greatly be reduced by merging some of the transition probabilities in a smart way. All the suboptimal schemes share one property that scheduling of at least one packet per state must be facilitated to maximize the satisfaction of continuity constraint before reaching $M^{\rm th}$ state. Note that the forced transmission in $M^{\rm th}$ time slot results in large energy expenditure.

\subsection{One-Or-All Scheduler}
This scheduler is the simplest in terms of computational complexity and implementation. Instead of having the option of scheduling up to $\mu+1$ packets in a state $p$, the user is limited to schedule either one, $\mu+1$ or no packet at all, thereby this scheme is called One-Or-All (OOA) scheduling. All other transition probabilities are set to zero. The idea is motivated by the emptying buffer scheme in \cite{majid_WCL12} where a user either empties the buffer when scheduled or waits for the next time slot.
The computational complexity for this scheme is $\mathcal{O}(B)$ for $B/N>>1$ as only $2B$ transition probabilities need to be optimized.

Following the derivation in (\ref{eqn:SVU_cdf3}), the cdf of the VUs for OOA is given by
\begin{equation}
{\rm \acute{P}}_{f,\rm VU}(y)=\sum\limits_{p=0}^M c_p\pi_p
\begin{cases}
(\mu+1){\rm P}_{f}(y)-\\
\mu{\rm P}_{f} (\kappa_{p0})-{\rm P}_{f}(\kappa_{p\mu})&y>\kappa_{p0}\\
{\rm P}_{f}(y)-{\rm P}_{f}(\kappa_{p\mu})&\kappa_{p\mu}< y\leq \kappa_{p0}\\
\end{cases}
\label{eqn:SVU_cdf4}
\end{equation}

\subsection{Selective State with Exponential Merging (SSE) }
As OOA scheduler quantizes every channel state into 2 levels per state, the performance is expected to degrade rapidly as compared to the "Best" one. A tradeoff between OOA and the "Best" scheduler is Selective State with Exponential merging (SSE) scheduler with complexity $\mathcal{O}(B\log B)$. Out of the possible state transitions, one state transition is dedicated for allowing scheduling of a single packet to maximize the probability of meeting the continuity constraint after dropping $N$ packets successively. Thus, the non-zero probability for the state transition from state $p$ to ${\mu}$ is a must for every state $p$. For the selection of other possible state transitions, we propose the following method where we choose thresholds exponentially.

For a state $p$, we observe in transition probability matrix for the "Best" scheme that the state transition probabilities are quite high for state zero, state $p$ and the states which are closer to zero. Therefore, other than $\alpha_{p\mu}$ (as in OOA), a natural choice of allowed transition probabilities is as follows.

For a state $p$, select a vector of possible next states by
\begin{equation}
\vec{q}=[2^0-1,2^1-1\dots2^{\lfloor\log_2(\mu+1)\rfloor}-1,\mu]
\label{eqn:q_vec}
\end{equation}
where the state transition from $p$ to ${\mu}$ is appended in the end only if it is not already contained by the vector, i.e., $2^{\lfloor\log_2(\mu+1)\rfloor}-1\ne\mu$.
The next states for the state transitions are more concentrated for the states with small $q$ and sparsity increases exponentially following our observation for the transition probability matrix for the "Best" scheme. Please note that smaller $q$ implies transmission of more packets and vice versa. The loss in merging states with large $q$ is small as compared to merging states with small $q$.

For this scheme, the corresponding state transition matrix $\mathbf{Q}_{\rm SSE}$ for a system with $N=2$ and $B=3$ is given by
\begin{equation}
\mathbf{Q}_{\rm SSE} = \left( \begin{array}{cccccc}
{\alpha}^s_{00} & {\alpha}^b_{01}&0  &0&0&0 \\
{\alpha}^s_{10} & {\alpha}^s_{11} & {\alpha}^b_{12} &0&0&0 \\
{\alpha}^s_{20} & {\alpha}^s_{21} &{\alpha}^s_{22}& {\alpha}^b_{23}&0&0\\
{\alpha}^s_{30} & {\alpha}^s_{31}  & 0 & {\alpha}^s_{33}&\alpha^d_{34}&0\\
{\alpha}^s_{40} & {\alpha}^s_{41}  & 0 & {\alpha}^s_{43}&0&{\alpha}^d_{45}\\
{\alpha}^s_{50} & {\alpha}^s_{51}  & 0& {\alpha}^s_{53}&0&0\
\end{array} \right) \label{eqn:Drop_STM_EM}.
\end{equation}

The general closed form expression for computation of VUs' channel distribution in terms of $p$ is not straight forward due to its dependency on vector $\vec{q}$. For every $p$ in matrix $\mathbf{Q}_{\rm SSE}$, the general form of the expression varies and can be derived using the derivation technique explained in Section \ref{sect:scheme} for the "Best" scheme. Thus, for every $p$, corresponding $\vec{q}$ is computed using (\ref{eqn:q_vec}) which determines the channel distribution as follows:

The channel distribution of VUs for SSE is conditioned on state $p$ (as before) and given by
\begin{equation}
{\rm \check{P}}_{f,\rm VU}(y)=\sum\limits_{p=0}^M c_p\pi_p {\rm P}_f(y|p)
\label{eqn:SVU_SS1}
\end{equation}
where ${\rm P}_f(y|p)$ in turns is conditioned on $\vec{q}$ such that:\newline
For $|\vec{q}|=1$,
\begin{equation}
{\rm P}_f(y|p)={\rm P}_{f}(y)-{\rm P}_{f}(\kappa_{p\mu})\quad y> \kappa_{p0}
\end{equation}
For $|\vec{q}|=2$,
\begin{equation}
{\rm P}_f(y|p)=\begin{cases}
(\mu+1)
{\rm P}_{f}(y)-{\rm P}_{f} (\kappa_{p0})&y>\kappa_{p0}\\-{\rm P}_{f}(\kappa_{p1})\\
{\rm P}_{f}(y)-{\rm P}_{f}(\kappa_{p\mu})& \kappa_{p\mu}< y\leq \kappa_{p0}\\
\end{cases}
\end{equation}
For $|\vec{q}|=3$,
\begin{equation}
{\rm P}_f(y|p)=\begin{cases}
(\mu+1){\rm P}_{f}(y)-{\rm P}_{f} (\kappa_{p0})&y>\kappa_{p0}\\-(\mu-1){\rm P}_{f} (\kappa_{p1})-{\rm P}_{f}(\kappa_{p\mu}) \\
\mu{\rm P}_{f}(y)-(\mu-1){\rm P}_{f} (\kappa_{p1})&\kappa_{p1}< y\leq \kappa_{p0}\\-{\rm P}_{f}(\kappa_{p\mu})\\
{\rm P}_{f}(y)-{\rm P}_{f}(\kappa_{p\mu})&\kappa_{p\mu}< y\leq \kappa_{p1}
\end{cases}
\end{equation}
Following the proposed framework, channel distribution can be calculated for all $|\vec{q}|,p$.

\section{Numerical Evaluation}
\label{sect:results}
In this section, we provide some numerical examples to demonstrate the potential gain of our scheme. $K$ users are uniformly distributed in a circular cell except a forbidden region of radius $\delta=0.01$ around the access point and the path loss follows the distribution in \cite{majid_TWC:13,Ralf1}. We assume Rayleigh fading with mean one and the path loss is exponential with exponent 2. The value of $C$ is fixed to $0.5$ bits/s/Hz.

First, we present an example of numerically computed $\mathbf{Q}_{\rm bst}$ matrix for the "Best" scheme. The matrix is a function of scheduling thresholds as computed in Appendix \ref{app:k_relation}. For example, we have the following matrix for $N=2,B=1$ and $\theta_{\rm tar}=0.02$.

\begin{equation}
\mathbf{Q}_{\rm bst} = \left( \begin{array}{cccc}
0.69& 0.31& 0& 0\\
0.76&0.17&0.07&0\\
0.67&0.27&0&0.06\\
0.64&0.36&0&0
\end{array} \right)
\label{eqn:Drop_STM1}
\end{equation}

\begin{figure}
    \centering
   \includegraphics[width=3.5in]{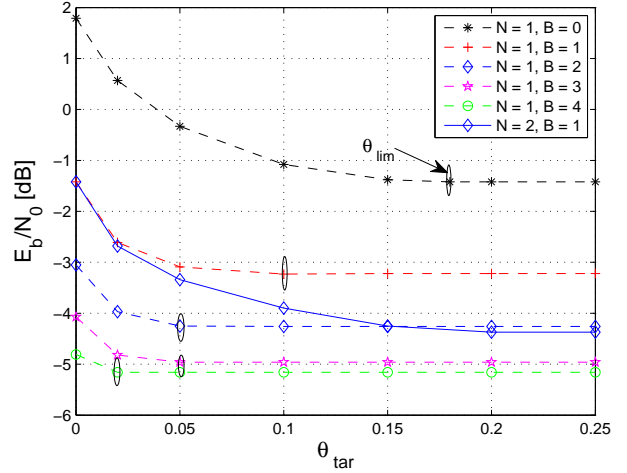}
     \caption{The system energy as a function of average packet dropping probability and buffer size parameters for a fixed $N$. The values for $\theta_{\rm lim}$ have been pointed with elliptical shapes for every curve.}
  \label{fig:system_energy}
  \vspace{-0.3cm}
\end{figure}

\begin{figure*}[!t]
 \centering
 \subfigure[$\theta_{\rm tar}=0$]
 {\includegraphics[width=3.5in]{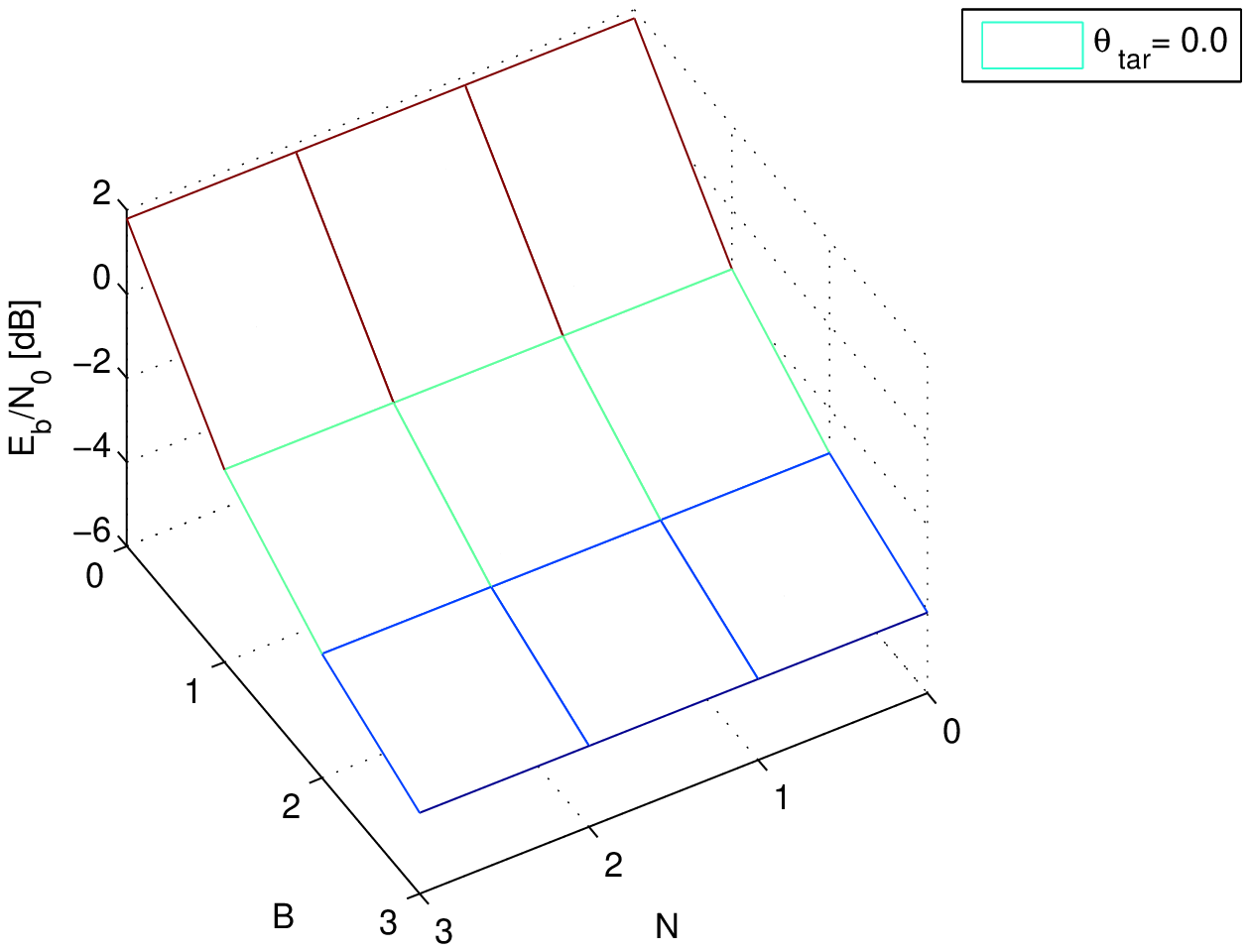}
 \label{fig:3d_a}}
\subfigure[$\theta_{\rm tar}=0.05$]
 {\includegraphics[width=3.5in]{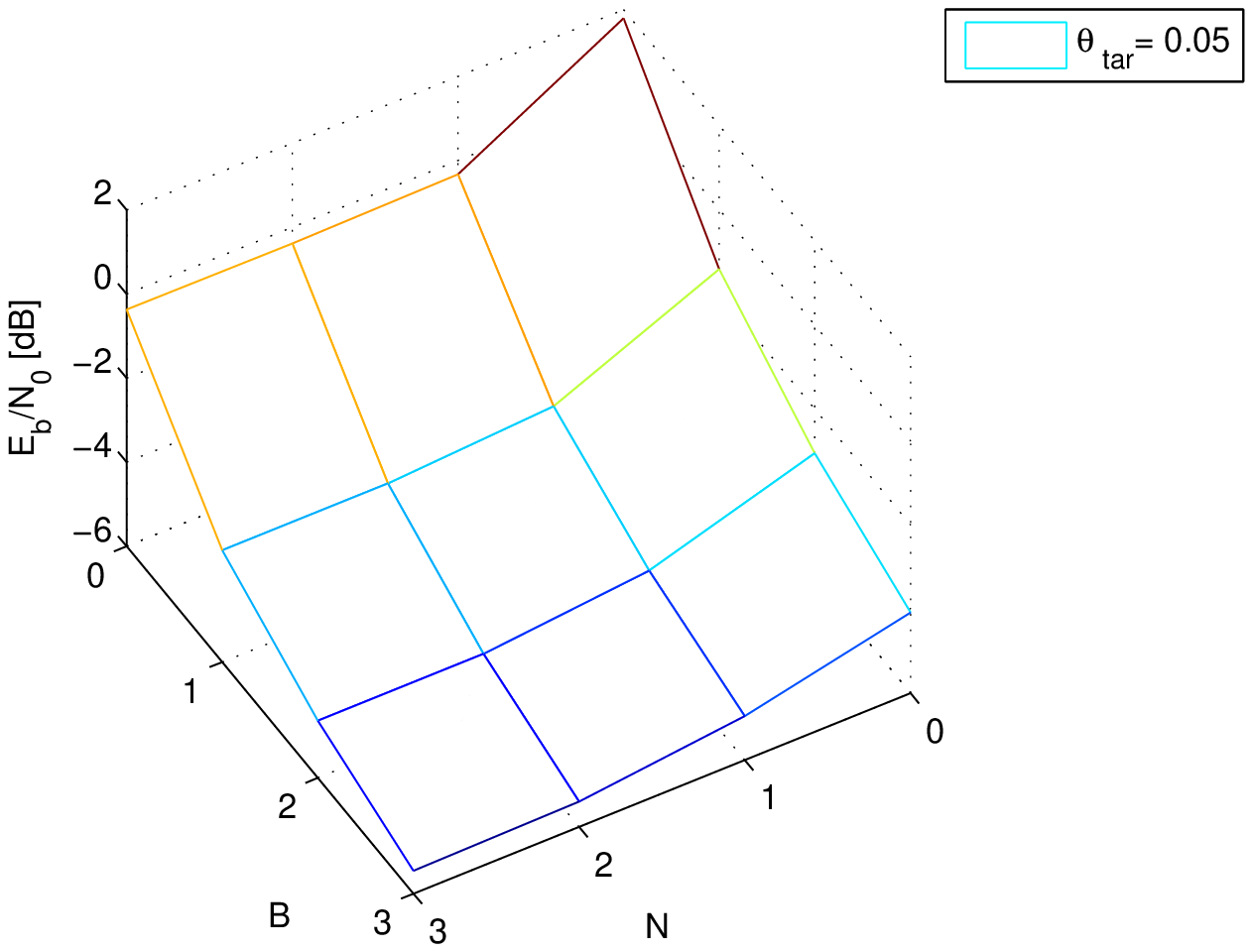}
 \label{fig:3d_b}}
  \caption{The system energy as a function of parameters $B,N$ and $\theta_{\rm tar}$ in a 3-dimensional illustrations.}
  \label{fig:3d_illustraton}
  \vspace{-0.3cm}
\end{figure*}

We show system energy against target average packet drop probability in Fig. \ref{fig:system_energy} and focus on the effects of buffering on a system constrained by parameters $\theta_{\rm tar},N$. For a fixed $N$, the system becomes more energy efficient as $B$ increases due to the effects characterized in Section \ref{sect:buffer_size}. For the case $B=1,N=1,\theta_{\rm tar}=0.02$, the parameter $\theta_{\rm tar}/\lambda_B$ equals $0.09$ from the matrix in (\ref{eqn:Drop_STM1}). This implies a system with buffer size one and $\theta_{\rm tar}=0.02$ is more energy efficient than a system with $B=0$ and $\theta_{\rm tar}=0.09$ as evident from Fig.~\ref{fig:system_energy}. Thus, flexibility in latency requirements helps to combat the transmission challenges emerging from the finite packet dropping parameters. We observe also that the energy expenditure for the parameters $B=2,N=1$ is substantially smaller than the case with parameters $B=1,N=2$ (where $M=3$ for both cases) at small $\theta_{\rm tar}$, but the opposite holds at large $\theta_{\rm tar}$. Thus, it is important to realize the operating region for the system to maximize the advantage from DoFs.

As evident from Fig. \ref{fig:system_energy}, an increase in value of $B$ results in decrease in the value of $\theta_{\rm lim}$ as explained in Lemma 1. In this particular numerical example, $B_m(N)=6$ for $N=1$. However, allowing $B=6$ is not beneficial from the energy point of view as no packet loss tolerance is exploited. If the system is loss tolerant in terms of average packet drop rate, $B<B_m(N)$ can be designed to maximize the energy efficiency without wasting extra buffer (cost). For example, for $\theta_{\rm tar}=0.05$ and $N=1$, $B=3$ realizes nearly the same energy as for the cases $B>3$ at the reduced cost of buffering.

Fig. \ref{fig:3d_illustraton} exhibits the effects of different combinations of parameter sets $B,N$ and $\theta_{\rm tar}$ on the system energy in a 3-dimensional plot where the effect of each parameter in different operating regions is evident. Fig.~\ref{fig:3d_a} for parameter $\theta_{\rm tar}=0$ is a special case where $N$ becomes irrelevant as zero average packet drop rate implies that the system is lossless and thus, $N>0$ does not help (as no packet can be dropped). However, if $\theta_{\rm tar}=0$, an increase in the value of $B$ does help to make the system energy efficient as shown in Fig. \ref{fig:3d_b}. A system with parameters $B=2,N=1, \theta_{\rm tar}=0$ is almost as energy efficient as a system with $B=1,N=1$ and $\theta_{\rm tar}\simeq0.05$, i.e., with an additional freedom in average dropping probability.
\begin{figure}
\centering
   \includegraphics[width=3.5in]{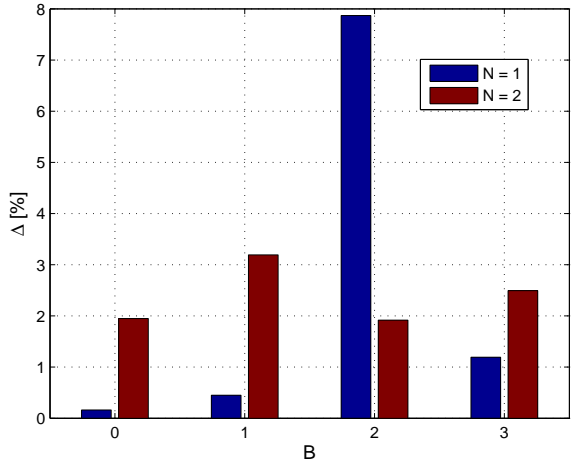}
  \caption{The accuracy measure $\Delta[\%]$ as a function of buffer size and continuity constraint parameters while $\theta_{\rm tar}=0.05$ for all simulations. The number of temperature iterations are $100$ while $50*(M+1)$ random configurations of $\mathbf{Q}$ are simulated at one temperature. The value of initial temperature $T_0$ can be fixed according to the problem requirements.}
  \vspace{-0.3cm}
  \label{fig:accuracy}
\end{figure}

To measure the relative accuracy of the computed solution for the SA algorithm, we define a parameter $\Delta$ by
\begin{equation}
\Delta=1-\frac{\theta_r^*}{\min(\theta_{\rm tar},\theta_{\rm lim})}
\end{equation}
where $\theta_{r}^*$ is computed for a given $\theta_{\rm tar}$ by using (\ref{eqn:drop_cons2}) for the "Best" solution $\mathbf{Q}^*$. This measure specifies how close $\theta_r^*$ for the solution to the targeted maximum value of average packet drop is. Though, it is not true always that the smaller $\Delta$ guarantees more energy efficiency as there exist some combinations of transition probabilities which result in smaller $\theta_r$ but larger energy. A small $\Delta$ implies that the computed solution is sufficiently exploiting the DoF inherited by the system through parameter $\theta_{\rm tar}$.

Fig.~\ref{fig:accuracy} shows $\Delta$ for different combinations of parameters $B$ and $N$. We observe that $\Delta$ is quite small for the computed solutions for the "Best" scheme. As SA is a heuristic algorithm, there is no consistent pattern in the values of $\Delta$. In general, for a fixed number of temperature iterations, the computation of the solution is expected to be hard as the number of parameters involved increases, i.e., dense $\mathbf{Q}$ and large number of thresholds.
If $\Delta$ is large for a parameter set, the number of iterations are increased to improve the solution.

\begin{figure*}[!t]
 \centering
 \subfigure[$E_b/N_0$ for the "Best" and OOA schemes for Rayleigh fading distribution.]
 {\includegraphics[width=3.5in]{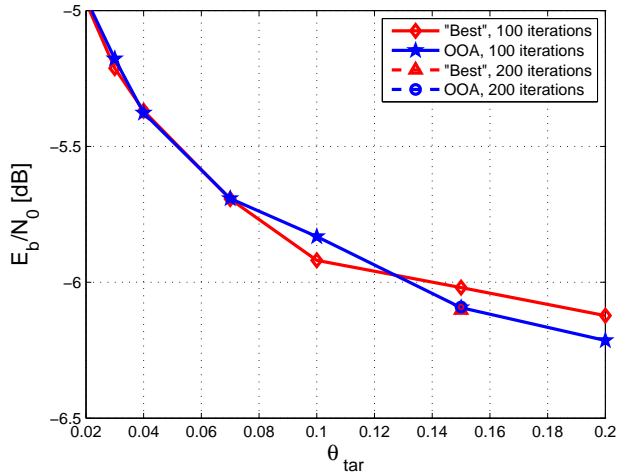}
 \label{fig:energy_comp}}
\subfigure[$\theta_r$ for the "Best" and OOA schemes for Rayleigh fading distribution.]
 {\includegraphics[width=3.5in]{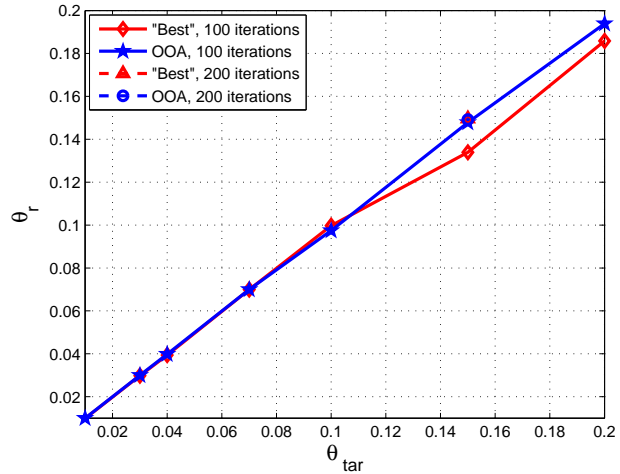}
 \label{fig:theta_comp}}
  \caption{The comparison of the "Best" scheme against suboptimal low complexity OOA scheme. The simulation parameters are $B=3, N=3$. Note that the curves for both schemes with 200 iterations per temperature value are overlapping with OOA curve with 100 iterations and are hardly distinguishable.}
  \label{fig:comparison}
  \vspace{-0.3cm}
\end{figure*}

Table \ref{comparison} illustrates the comparison between the 'Best' and low complexity sub optimal schemes. Note that all the schemes fulfill the continuity and average packet drop constraints effectively. As expected, OOA is the least complex scheme but energy loss is substantial as compared to the "Best" one at $\theta_{\rm tar}=0.10$. The SSE scheme is more complex than OOA; but improves energy performance over OOA while energy loss is minimal as compared to the "Best" one.

\begin{table}[t!]
\renewcommand{\arraystretch}{1.0}
\caption{Comparison of Schemes for $N=3,B=3,C=0.5$ bits/s/Hz}
\label{comparison}
\begin{center}
\begin{tabular}{|c||c|c|}
\hline
Scheme & Complexity &$E_b/N_0\big(\theta_{\rm tar}=0.1\big)$\\
\hline
"Best"&$\mathcal{O}(B^2)$&-5.92 dB\\
\hline
OOA&$\mathcal{O}(B)$&-5.83 dB\\
\hline
SSE&$\mathcal{O}(B\log B)$&-5.86 dB\\
\hline
\end{tabular}
\end{center}
\vspace{-0.3cm}
\end{table}

Fig. \ref{fig:comparison} compares the system energy and achieved average packet drop rate $\theta_r$ for the "Best" and OOA schemes against a target average dropping probability. At low $\theta_{\rm tar}$ both schemes are indistinguishable as DoF provided by $\theta_{\rm tar}$ for $N=3$ and $B=3$ is too small to be capitalized effectively by any scheme. As $\theta_{\rm tar}$ increases, the "Best" scheme shows better results than OOA as expected. However, for large $\theta_{\rm tar}$, OOA outperforms the "Best" one surprisingly. This is attributed to the fact that reduced complexity (less number of thresholds) of OOA helps to compute more accurate matrix $\mat{\rm Q}_{\rm OOA}^*$ (via SA) as compared to the "Best" scheme where $\mat{\rm Q}_{\rm bst}^*$ is dense. This effect is more pronounced at large $\theta_{\rm tar}$ due to more freedom in choosing transition probabilities in matrix $\mat{\rm Q}^*$.

To verify our observations numerically, we repeat the optimization by increasing number of iterations per single temperature from 100 to 200. We also plot the achieved $\theta_r$ for both schemes in Fig.~\ref{fig:theta_comp}. Initially, $\theta_r$ follows $\theta_{\rm tar}$ closely for both schemes. It is clear that $\theta_r$ for OOA is much closer to $\theta_{\rm tar}$ in the region when OOA outperforms the "Best" scheme. At $\theta_{\rm tar}=0.15$, increase in number of iterations in SA algorithm from 100 to 200 improves accuracy of $\mat{\rm Q}^*$ (and thus system energy in Fig.~\ref{fig:energy_comp}) for the "Best" scheme while it has negligible effect on OOA scheme as $\theta_r$ for OOA was already close to $\theta_{\rm tar}$. Thus, we conclude that the proposed low complexity schemes perform very close to the "Best" scheme and extensive fading vector quantization levels for the "Best"  scheme do not help much.

\section{Conclusion}
\label{sect:conclusions}
We investigate the tradeoff between the system energy of multiuser multi-access system and the packet drop tolerance of the applications characterizing the network traffic. In contrast to common approach of dropping a packet as a consequence of failing to provide a required rate to the users, we propose maximizing the use of packet drop tolerance by dropping as many packets as permissible without compromising the QoE for the users. The joint constraint on permissable average and successive packet drop poses interesting challenge in the optimization problem. We propose a packet scheduling scheme and analyze it using MDP under large user limit. As the formulated optimization problem is non-convex for a multiuser system, the heuristic solution is presented. We also propose suboptimal low complexity schemes which show negligible energy loss as compared to the proposed "Best" scheme. The numerical results evaluate trade-offs between the system energy and QoS parameters. The study reveals that system energy is influenced by different parameters in different operating regions and it is important to quantify the effect of each parameter to opportunistically make use of the channel for an energy efficient system design. In future work, we consider the proposed framework in more realistic scenarios when CSI is not available and only statistical guarantees can be provided on successive packet loss.

\appendices
\section{Relation between Thresholds and Transition Probabilities}
\label{app:k_relation}
There is one to one mapping between transition probabilities and thresholds. For a given set of parameters $B,N$ and small scale fading distribution ${\rm P}_{f}(y)$, the
transition probabilities are a function of the scheduling thresholds. Therefore, computing a set of
thresholds is equivalent to the computation of a set of transition probabilities.

For a fixed $M=B+N$ (assuming $B\geq 1$),
the transition probability matrix $\mathbf{Q}_{\rm bst}$ for the "Best" scheme in is expressed as,
\begin{eqnarray}
 \mathbf{Q}_{\rm bst}&=&\left(
\begin{array}{cccc}
{\rm Pr}(y\geq \kappa_{00})& {\rm Pr}(y< \kappa_{00}) & \cdots &0\\
{\rm Pr}(y\geq \kappa_{1 0})& {\rm Pr}(\kappa_{11}\leq y< \kappa_{10}) &\cdots &0 \nonumber\\
\ddots & \ddots & \ddots &\ddots\\
{\rm Pr}(y\geq \kappa_{M0})& {\rm Pr}(\kappa_{M1}\leq y<
\kappa_{M0}) &\cdots &0\\
\end{array}
\right)\\
\end{eqnarray}
where zero transition probability represents the impossible transition.

\section{Channel Modeling}
\label{app:channel}
In this work the channel model of \cite{majid_TWC:13,Ralf1} is used. Signal
propagation is characterized by a distance dependent path loss
factor and a frequency-selective short-term fading that depends on
the scattering environment around the user terminal. As described in
Section \ref{sect:system model}, these two effects are taken into
account by letting $h_k=s_k f_k$.

As in \cite{Ralf1}, we assume that users are uniformly distributed
in a geographical area but for a forbidden circular region of radius
$\delta$ centered around the base station where $0<\delta\leq 1$ is
a fixed system constant. Using this model, the cdf of path loss is
given by
\begin{equation}
{\rm F}_s(x) = \left\{
\begin{array}{c@{\qquad \qquad}c}
0& x < 1 \\
1-\frac{x^{-2/\alpha}-\delta^2}{1-\delta^2} &  1\leq x < \delta^{-\alpha}\\
1&x\geq \delta^{-\alpha}
\end{array}
\right. \label{eqn:path_loss}.
\end{equation}
where the path loss at the cell border is normalized to one.

For a Rayleigh channel, the distribution of small scale block fading is given by
\begin{equation}
{\rm P}_{f}(y)=1-\exp(-y) \label{eqn:fading}
\end{equation}
${\rm P}_{h}(x)$ is defined as the cdf of the random
variable $h_{k}=s_kf_{k}$. Recall from (\ref{eqn:SVU_cdf3}), the cdf
${\rm P}_{f,{\rm VU}}(y)$ of VUs for the "Best" scheme is a weighted function
of the cdf of actual fading ${\rm P}_{f}(y)$ given by (\ref{eqn:fading}). Using (\ref{eqn:SVU_cdf3}) and (\ref{eqn:path_loss}), we compute a convenient expression for the
cdf ${\rm P}_{h,{\rm VU}}(x)$ of the VUs for this product channel.
As path loss and Rayleigh fading occur simultaneously and
independently, the cdf of the channel gain is given by
\begin{equation}
{\rm P}_{h,{\rm VU}}(x)=\int F_s(x/y){\rm dP}_{f,{\rm VU}}(y).
\label{eqn:cdf_channel}
\end{equation}
Using path loss distribution in (\ref{eqn:path_loss}), (\ref{eqn:cdf_channel}) is computed as
follows
\begin{eqnarray}
{\rm P}_{h,{\rm VU}}(x)= \int_{0}^{x\delta^\alpha}{\rm p}_{f,\rm{
VU}}
(y){{\rm d}y}+\int_{x\delta^\alpha}^x F_s(x/y){\rm dP}_{f,{\rm VU}} (y)
 \label{eqn:cdf_channel_1}
\end{eqnarray}
Following the derivation in \cite{majid_TWC:13}, changing variables and integrating by parts yields,
\begin{equation}
{\rm P}_{h,{\rm VU}}(x)=
\frac{1}{x^{2/\alpha}(1-\delta^2)}\int_{x^{2/\alpha}
\delta^2}^{x^{2/\alpha}}{\rm P}_{f,{\rm VU}} (y^{\alpha/2}){\rm
{dy}}~. \label{eqn:arraycdf_channel_11}
\end{equation}
For $\alpha=2$, (\ref{eqn:arraycdf_channel_11}) can be written
in closed form.

For the "Best" scheduler, using (\ref{eqn:SVU_cdf3}) and the
Rayleigh fading model, (\ref{eqn:arraycdf_channel_11}) is given
by
\begin{eqnarray}
\label{eqn:cdf_product1}
{\rm P}_{h,{\rm VU}}(x)&&= \frac{1}{x(1-\delta ^2)}\int_{x
\delta^2}^{x}\sum\limits_{p=0}^M c_p\pi_p\Bigl[
L(p,y)\\&&(1-\exp(-y))
-\sum_{b=0}^{\mu-q}(1-\exp(-\kappa_{p(\mu-b)})) \Bigr] \rm{dy}~.\nonumber
 \label{eqn:cdf_product1}
\end{eqnarray}
Integrating (\ref{eqn:cdf_product1}) yields
\begin{eqnarray}
&&{\rm P}_{h,\rm {VU}}(x)=\sum\limits_{p=0}^M c_p\pi_p
\Bigl[L(p,y)\Bigl(1+\frac{1}{x(1-\delta^2)}\bigl(\exp(-x)\nonumber\\
&&-\exp(-x\delta^2)\bigr)\Bigr)
-\sum_{b=0}^{\mu-q}\bigl(1-\exp(-\kappa_{p
(\mu-b)})\bigr)\Bigr]~.
\label{eqn:cdf_product2}
\end{eqnarray}
Similarly, the channel distribution of VUs for OOA and SSE schedulers can be computed using (\ref{eqn:SVU_cdf4}) and (\ref{eqn:SVU_SS1}), respectively in (\ref{eqn:arraycdf_channel_11}).

\section{Review of Lemma in \cite{majid_TWC:13}}
\label{app:review}
The behaviour of system energy for a fixed $N$ and large $\theta_{\rm tar}$ has been characterized by a Lemma in \cite{majid_TWC:13}. We review the Lemma briefly here to keep this work independent and self-sufficient.
The lemma states:

For a fixed continuity constraint parameter $N$, there exists a finite $\theta_{\rm lim}$ such that for all $\theta_{\rm tar} > \theta_{\rm lim}$, the same maximum energy efficiency is achieved as for more restrictive $\theta_{\rm lim}$.

We evaluate (\ref{eqn:drop_cons2}) from $\mathbf{Q}^*$ to get the resulting average dropping probability $\theta_r^*$. For a fixed $N$, we will not be able to achieve further energy efficiency by dropping more packets for any $\theta_{\rm tar}>\theta_{\rm lim}$. The $\theta_r^*$ value for the special case is termed as \emph{limiting} average dropping probability and denoted by $\theta_{\rm lim}$ when we have no average packet dropping constraint but only continuity constraint. Equivalently, $\theta_{\rm lim}$ is the solution of (\ref{eqn:objective}) without applying $\mathcal{C}_2$ in (\ref{eqn:optimization}).

\bibliographystyle{IEEEtran}
\bibliography{bibliography}

\begin{thebibliography}{10}
\providecommand{\url}[1]{#1}
\csname url@samestyle\endcsname
\providecommand{\newblock}{\relax}
\providecommand{\bibinfo}[2]{#2}
\providecommand{\BIBentrySTDinterwordspacing}{\spaceskip=0pt\relax}
\providecommand{\BIBentryALTinterwordstretchfactor}{4}
\providecommand{\BIBentryALTinterwordspacing}{\spaceskip=\fontdimen2\font plus
\BIBentryALTinterwordstretchfactor\fontdimen3\font minus
  \fontdimen4\font\relax}
\providecommand{\BIBforeignlanguage}[2]{{%
\expandafter\ifx\csname l@#1\endcsname\relax
\typeout{** WARNING: IEEEtran.bst: No hyphenation pattern has been}%
\typeout{** loaded for the language `#1'. Using the pattern for}%
\typeout{** the default language instead.}%
\else
\language=\csname l@#1\endcsname
\fi
#2}}
\providecommand{\BIBdecl}{\relax}
\BIBdecl

\bibitem{majid_PIMRC_13}
M.~M. Butt and E.~A. Jorswieck, ``Maximizing energy efficiency for loss
  tolerant applications: The packet buffering case,'' in \emph{IEEE
  International Symposium on Personal, Indoor and Mobile Radio Communications
  (PIMRC)}, London, UK, Sep 2013.

\bibitem{Zia:2011}
Z.~Hasan, H.~Boostanimehr, and V.~Bhargava, ``Green cellular networks: A
  survey, some research issues and challenges,'' \emph{IEEE Communications
  Surveys Tutorials}, vol.~13, no.~4, pp. 524--540, 2011.

\bibitem{Fehske:11}
A.~Fehske, G.~Fettweis, J.~Malmodin, and G.~Bicz\'{o}k, ``The global footprint
  of mobile communications: The ecological and economic perspective,''
  \emph{IEEE Communications Magazine}, vol.~49, no.~8, pp. 55--62, August 2011.

\bibitem{Han:2013}
T.~Han and N.~Ansari, ``On greening cellular networks via multicell
  cooperation,'' \emph{IEEE Wireless Communications}, vol.~20, no.~1, pp.
  82--89, 2013.

\bibitem{fu:06}
A.~Fu, E.~Modiano, and J.~N. Tsitsiklis, ``Optimal transmission scheduling over
  a fading channel with energy and deadline constraints,'' \emph{IEEE
  Transactions on Wireless Communications}, pp. 630--641, March 2006.

\bibitem{majid_WCL12}
M.~M. Butt, ``Energy-performance trade-offs in multiuser scheduling: Large
  system analysis,'' \emph{IEEE Wireless Communications Letters}, vol.~1,
  no.~3, pp. 217--220, June 2012.

\bibitem{Neely}
M.~J. Neely, ``Optimal energy and delay tradeoffs for multiuser wireless
  downlinks,'' \emph{IEEE Trans. Inform. Theory}, vol.~53, no.~9, pp.
  3095--3113, September 2007.

\bibitem{Gallager}
R.~A. Berry and R.~G. Gallager, ``Communication over fading channels with delay
  constraints,'' \emph{IEEE Trans. Inform. Theory}, vol.~48, no.~5, pp.
  1135--1149, May 2002.

\bibitem{Huang:2004}
V.~Huang and W.~Zhuang, ``{QoS}-oriented packet scheduling for wireless
  multimedia {CDMA} communications,'' \emph{IEEE Transactions on Mobile
  Computing}, vol.~3, no.~1, pp. 73--85, 2004.

\bibitem{Neely2009}
M.~J. Neely, ``Intelligent packet dropping for optimal energy-delay tradeoffs
  in wireless downlinks,'' \emph{IEEE Trans. on Automatic Control}, vol.~54,
  no.~3, pp. 565--579, March 2009.

\bibitem{Karmokar}
A.~Karmokar, D.~Djonin, and V.~Bhargava, ``Optimal and suboptimal packet
  scheduling over correlated time varying flat fading channels,'' \emph{IEEE
  Transactions on Wireless Communications}, vol.~5, no.~2, pp. 446--456, 2006.

\bibitem{Bettesh:2006}
I.~Bettesh and S.~Shamai, ``Optimal power and rate control for minimal average
  delay: The single-user case,'' \emph{Information Theory, IEEE Transactions
  on}, vol.~52, no.~9, pp. 4115--4141, Sept 2006.

\bibitem{majid_TWC:13}
M.~Butt and E.~A. Jorswieck, ``Maximizing system energy efficiency by
  exploiting multiuser diversity and loss tolerance of the applications,''
  \emph{IEEE Transactions on Wireless Communications}, vol.~12, no.~9, pp.
  4392--4401, 2013.

\bibitem{Whittle:1988}
P.~Whittle, ``Restless bandits: Activity allocation in a changing world,''
  \emph{Journal of Applied Probability}, vol.~25, pp. pp. 287--298, 1988.

\bibitem{Lee:1994}
C.~Lee and M.~S. Andersland, ``Minimizing consecutive packet loss in real-time
  {ATM} sessions,'' in \emph{IEEE Global Telecommunications Conference},
  November 1994.

\bibitem{Lee_pat:1997}
C.~W.~L. Coralville and M.~S. Andersland, ``Control of consecutive packet loss
  in a packet buffer,'' Patent US 5\,629\,936, May 13, 1997.

\bibitem{Fanqqin:2013}
F.~Liu, T.~H. Luan, X.~S. Shen, and C.~Lin, ``Dimensioning the packet loss
  burstiness over wireless channels: a novel metric, its analysis and
  application,'' \emph{Wireless Communications and Mobile Computing}, 2012.

\bibitem{Ralf1}
G.~Caire, R.~M\"{u}ller, and R.~Knopp, ``Hard fairness versus proportional
  fairness in wireless communications: The single-cell case,'' \emph{IEEE
  Trans. Inform. Theory}, vol.~53, no.~4, pp. 1366--1385, April 2007.

\bibitem{Forney:1998}
J.~Forney, G.D. and G.~Ungerboeck, ``Modulation and coding for linear gaussian
  channels,'' \emph{IEEE Transactions on Information Theory}, vol.~44, no.~6,
  pp. 2384--2415, Oct 1998.

\bibitem{Tse3}
D.~Tse and S.~Hanly, ``Multi-access fading channels-part {I}: Polymatroid
  structure, optimal resource allocation and throughput capacities,''
  \emph{IEEE Trans. Inform. Theory}, vol.~44, no.~7, pp. 2796--2815, November
  1998.

\bibitem{viswanath:01}
P.~Viswanath, D.~N. Tse, and V.~Anantharam, ``Asymptotically optimal
  water-filling in vector multiple-access channels,'' \emph{IEEE Trans. Inform.
  Theory}, vol.~47, no.~1, pp. 241--267, January 2001.

\bibitem{jindal_2009}
J.~Lee and N.~Jindal, ``Energy-efficient scheduling of delay constrained
  traffic over fading channels,'' \emph{IEEE Trans. Wireless Communications},
  vol.~8, no.~4, pp. 1866--1875, April 2009.

\bibitem{BA}
S.~Geman and D.~Geman, ``Stochastic relaxation, gibbs distribution and the
  baysian restoration in images,'' \emph{IEEE transactions on pattern analysis
  and machine intelligence}, vol.~6, no.~6, pp. 721--741, November 1984.

\bibitem{FA}
H.~Szu and R.~Hartley, ``Fast simulated annealing,'' \emph{Physics Letters A},
  vol. 122, no.~3, 1987.

\bibitem{SA_1}
S.~Kirkpatrick, C.~Gelatt, and M.~Vecchi, ``Optimization by simulated
  annealing,'' \emph{Science}, vol. 220, no. 4598, pp. 671--680, May 1983.

\bibitem{SA_2}
V.~Cerny, ``Thermodynamical approach to the travelling salesman problem: An
  efficient simulation algorithm,'' \emph{Journal of Optimization Theory and
  Applications}, vol.~45, no.~1, pp. 41--52, January 1985.

\end{thebibliography}

\end{document}